\pdfoutput=1
\documentclass[submission,copyright,creativecommons]{eptcs}
\providecommand{\event}{AFL 2026} 

\usepackage{iftex}

\ifpdf
\usepackage{underscore}         
\usepackage[T1]{fontenc}        
\else
\usepackage{breakurl}           
\fi
\newif\iflong
\longtrue
\usepackage{xcolor}
\usepackage{amsmath,amssymb,amsthm}
\usepackage{MnSymbol}
\usepackage{tikz}
\usetikzlibrary{calc}
\newtheorem{theorem}{Theorem}
\newtheorem{lemma}[theorem]{Lemma}
\newtheorem{proposition}[theorem]{Proposition}
\newtheorem{corollary}[theorem]{Corollary}
\theoremstyle{definition}
\newtheorem{definition}{Definition}
\theoremstyle{remark}
\newtheorem{remark}{Remark}

\title{\textbf{Guillotine and Tiling Cofiniteness \\in Unary Picture Languages}}
\author{Pierluigi {San Pietro}$^{[0000-0002-2437-8716]}$  \quad  Stefano {Crespi Reghizzi}$^{[0000-0001-5061-7402]}$
    \institute{Dipartimento di Elettronica, Informazione e Bioingegneria}
    \institute{Politecnico di Milano and CNR-IEIIT \\
        Milano, Italia}
    \email{\quad pierluigi.sanpietro@polimi.it\quad\qquad stefano.crespireghizzi@polimi.it}
    \and
    Antonio Restivo$^{[0000-0002-1972-6931]}$
    \institute{Dipartimento di Matematica e Informatica \\ Universit\`{a} di Palermo \\
        Palermo, Italia}
    \email{ antonio.restivo@unipa.it }
}

\newcommand{\titlerunning}{Guillotine and Tiling Cofiniteness in Unary Picture Languages}
\newcommand{\authorrunning}{P. San Pietro, S. Crespi Reghizzi, A. Restivo}

\hypersetup{
    bookmarksnumbered,
    pdftitle    = {\titlerunning},
    pdfauthor   = {\authorrunning},
    pdfsubject  = {EPTCS},               
}

\begin{document}
	\maketitle
\begin{abstract}
Motivated by the study of closure operations on picture languages, we ask when a set of unary tiles can generate all sufficiently large pictures either via  repeated horizontal and vertical concatenations or via tiling operations.

In one dimension, the lengths of words in the concatenation closure of a unary
language form an additive subsemigroup of \(\mathbb N_{>0}\); hence, they are finitely
generated. Cofiniteness is characterized by the gcd of these generators being
one. Since ordinary cofiniteness in two dimensions is too restrictive and
essentially reduces to degenerate one-dimensional conditions on strips, we
introduce and study  \emph{asymptotic cofiniteness}: all sufficiently large
rectangular pictures are generated. A direct extension of the one-dimensional
criterion, requiring the gcd of the tile heights and of the tile widths to be
one, is not sufficient, because local congruence obstructions may persist in two
dimensions.

We give an exact arithmetic characterization of asymptotic cofiniteness
for arbitrary, possibly infinite, sets of unary rectangular tiles.
The characterization is the same for the guillotine closure, obtained by horizontal and vertical concatenation, 
and for the full tiling closure.
For finite tile sets, the proof combines semigroup arguments and an explicit
least-common-multiple stacking construction for sufficiency, while necessity is
obtained by a roots-of-unity argument applying also to nonsliceable tilings. The
extension to infinite tile sets follows from the finite-basis theorem for Klarner
systems.
\end{abstract}

\section{Introduction}\label{sect:intro}

In one dimension, cofiniteness is a well understood property. 
If \(L\) is a finite language of unary words, then its concatenation closure \(L^{+}\) is cofinite if and only if the greatest common divisor of the lengths of the words occurring in \(L\) is \(1\). 

In this paper we study the corresponding question in two dimensions, namely for unary picture languages~\cite{GiammRestivo1997, BMV2007,DBLP:books/ems/21/Crespi-ReghizziGL21}. A picture language is unary when it is defined over a one-letter alphabet; in this case, every picture is a rectangular array of copies of the same symbol and is therefore completely determined by its height and width.
Although in the unary setting pictures can be identified with their rectangular sizes, the objects of concern are still picture languages and language-theoretic closure operations. 
A picture language generates larger pictures either by repeated horizontal and vertical concatenation, giving the \emph{guillotine closure}, or by arbitrary tilings, giving the \emph{tiling closure}~\cite{DBLP:journals/tcs/Simplot99}. 
We ask when all sufficiently large rectangular unary pictures belong to the generated language,  that is, when the set of generated sizes contains a northeast quadrant. We call this property \emph{asymptotic cofiniteness} 
(Figure~\ref{fig:cofiniteness-notions}, right).

These two closure operations are classical in the theory of picture languages. 
Arbitrary tilings play a central role in local and grammatical descriptions of pictures~\cite{DBLP:journals/tcs/Simplot99,DBLP:journals/tcs/CherubiniCPP06,DBLP:journals/iandc/PradellaCC11}, whereas closure under horizontal and vertical concatenation is more restrictive, since pictures are built recursively by composing two pictures with matching dimensions along a full common edge~\cite{GiammRestivo1997,DBLP:conf/stacs/Matz97}.
 
For recognizable picture languages~\cite{GiammarresiR92,DBLP:journals/tcs/LatteuxS97}, Simplot proved that closure under horizontal and vertical concatenation is strictly contained in closure under arbitrary tilings~\cite{DBLP:journals/tcs/Simplot99}. 
More generally, several familiar one-dimensional properties change substantially when passing to pictures: for instance, closure-theoretic decision problems that are decidable for word languages become undecidable already for local and recognizable picture languages, while bases of closed picture languages remain unique but need not be recognizable~\cite{DBLP:conf/dlt/CrespiReghizziRP26}. 
The present paper focuses on a different aspect of this contrast, namely asymptotic cofiniteness.

This asymptotic notion is natural in the present setting. 
Indeed, ordinary cofiniteness (Figure~\ref{fig:cofiniteness-notions}, left) in \(\mathbb N_{>0}^2\) would already force all but finitely many horizontal strips \((1,w)\) and vertical strips \((h,1)\) to belong to the generated language. 
But a rectangle of height \(1\) can only be generated from pieces of height \(1\), and dually for width \(1\). 
In effect, cofiniteness in the classical sense would reduce the problem to two independent one‑dimensional strips — a much less interesting situation.
In applications such as architectural floorplans or VLSI layouts~\cite{DBLP:journals/tcs/KantH97,DBLP:journals/siamcomp/EppsteinMSV12,DBLP:journals/tcs/KumarS21,DBLP:journals/combtheory/AsinowskiCFF25}, relying on arbitrarily long one-cell-wide corridors is often unnatural. 
This motivates asymptotic cofiniteness, which avoids such degenerate constructions and captures the ability to fill all large rectangles, while allowing infinitely many missing sizes (see Figure~\ref{fig:cofiniteness-notions}).

\begin{figure}[t]
    \centering
    \begin{tikzpicture}[scale=0.55,line cap=round,line join=round]
        \begin{scope}
            \draw[->] (0,0) -- (6.5,0) node[right] {$w$};
            \draw[->] (0,0) -- (0,5.5) node[above] {$h$};
            
            \foreach \x in {1,...,6}{
                \foreach \y in {1,...,5}{
                    \fill (\x,\y) circle (1.8pt);
                }
            }
            
            \foreach \p in {(1,1),(2,3),(4,2)}{
                \fill[white] \p circle (2.4pt);
                \draw \p circle (2.4pt);
            }
            
            \node at (3.3,-0.9) {\small ordinary cofiniteness:};
            \node at (3.3,-1.9) {\small only finitely many missing sizes $\circ$;};
        \end{scope}
        
        \begin{scope}[xshift=15cm]
            \draw[->] (0,0) -- (6.5,0) node[right] {$w$};
            \draw[->] (0,0) -- (0,5.5) node[above] {$h$};
            
            \draw[dashed] (2.5,0) -- (2.5,5.3);
            \draw[dashed] (0,2.5) -- (6.3,2.5);
            \node[below] at (2.5,0) {$w_0$};
            \node[left] at (0,2.5) {$h_0$};
            
            \fill[gray!15] (2.5,2.5) rectangle (6.3,5.3);
            
            \foreach \x in {1,...,6}{
                \foreach \y in {1,...,5}{
                    \fill (\x,\y) circle (1.8pt);
                }
            }
            
            \foreach \p in {(2,1),(3,1),(4,1),(5,1),(6,1),(1,3),(1,4),(1,5),(2,4),(2,5)}{
                \fill[white] \p circle (2.4pt);
                \draw \p circle (2.4pt);
            }
            
            \node at (3.3,-1) {\small asymptotic cofiniteness:};
             \node at (3.3,-2) {\small the shaded area contains  no missing size $\circ$;};
             \node at (3.3,-3) {\small outside the shaded area there  may be infinitely many $\circ$ };
        \end{scope}
    \end{tikzpicture}
    \caption{Ordinary cofiniteness allows only finitely many missing sizes (denoted by $\circ$), whereas
        asymptotic cofiniteness requires all sufficiently large heights and
        widths to be present simultaneously, but infinitely many sizes may still be missing.}
    \label{fig:cofiniteness-notions}
\end{figure}
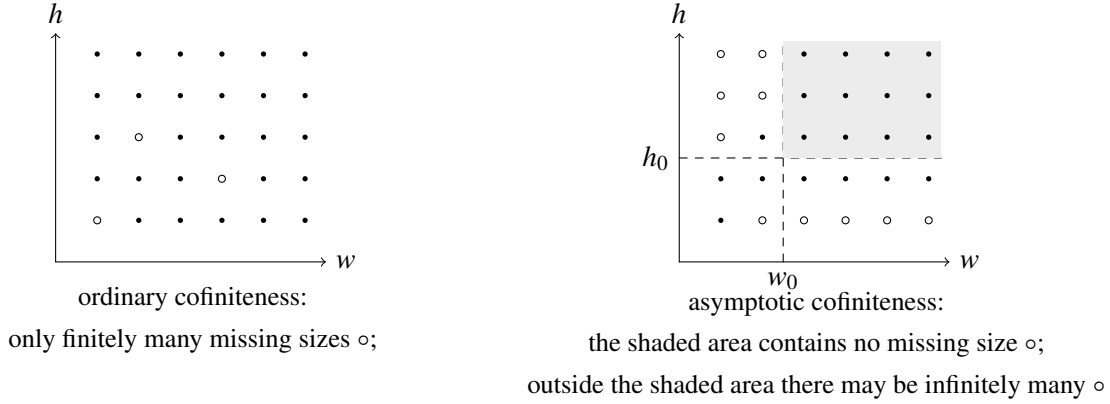

It is also useful to distinguish our setting from ordinary additive semigroups in \(\mathbb N^2\). 
If a finite set \(F\subseteq\mathbb N_{>0}^2\) generates a monoid under vector
addition, then all generated points lie in the rational cone determined by the
extreme slopes of the generators.
Hence even asymptotic cofiniteness cannot occur, since an upper-right quadrant contains points of arbitrarily small and arbitrarily large slope. 
Thus ordinary vector addition cannot produce the product‑shaped sets we need; the extra power of independent horizontal and vertical growth is essential.

Our viewpoint is also different from the classical two-dimensional packing and strip-packing problems~\cite{iori2021exact,DBLP:journals/dam/LodiMP17}, where the input is typically a finite set of rectangular items and the question is whether they can be packed into a prescribed container, or how to optimize a quantity such as the number of bins or the used height. 
Here, by contrast, the  tile set acts as a generating set, and the object of study is the infinite family of rectangles obtainable under closure operations, with cofiniteness as the relevant asymptotic property.

Returning to asymptotic cofiniteness, one might at first expect a direct extension of the one-dimensional criterion: perhaps it is enough that the gcd of the tile heights be \(1\) and the gcd of the tile widths be \(1\). 
This is false in general. 

For instance, following the convention that the first coordinate is height and
the second is width, we write \(a^{h,w}\) for the unary picture with \(h\)
rows and \(w\) columns. Thus \(a^{2,3}\) has height \(2\) and width \(3\).
Consider
\[
L=\{a^{2,3},a^{2,5},a^{7,3}\}.
\]

Although the gcd of the occurring heights is \(\gcd\{2,7\}=1\) and the gcd of the occurring widths is \(\gcd\{3,5\}=1\), every picture in $L^{++}$ of odd height has width divisible by \(3\). 
Hence asymptotic cofiniteness fails. 
Indeed, the only tiles of odd height have width \(3\), and the same obstruction propagates under tiling.

Our main result gives an exact arithmetic characterization, and shows that it is the same for the guillotine closure and for the tiling closure. 
Identifying unary pictures with their sizes, let \(F\subseteq\mathbb N_{>0}^2\) be the set of tile sizes. 
The precise condition involves only a few simple sets derived from 
$F$:
\[
H=\{\,h\ge 1 : \exists w\ge 1,\ (h,w)\in F\,\},
\]
\[
P=\{\, \text{prime } p: p\mid w \text{ for some } (h,w)\in F\,\},
\]
and, for each \(p\in P\),
\[
H_{p}=\{\,h\ge 1 : \exists (h,w)\in F \text{ with } p\nmid w\,\},
\]
where $p\mid w$ denotes that $p$ divides $w$ (i.e., $w$ is a multiple of $p$), and $p\nmid w$ otherwise.

Thus \(H_{p}\) is the set of heights of tiles whose width is not divisible by \(p\).

We prove that the following conditions are equivalent:
\begin{enumerate}
    \item the guillotine closure of \(F\) is asymptotically cofinite;
    \item the tiling closure of \(F\) is asymptotically cofinite;
    \item \(\gcd(H)=1\) and \(\gcd(H_{p})=1\) for every \(p\in P\).
\end{enumerate}

We first prove the characterization for finite tile sets. 
The proof is split into two parts. 
For the guillotine closure, the sufficiency direction $(3) \Longrightarrow (1)$ combines semigroup arguments with an explicit least-common-multiple stacking construction. 
The necessity direction $(2) \Longrightarrow (3)$ is then proved simultaneously for the guillotine and tiling closures by a roots-of-unity argument that works even for nonsliceable tilings such as windmills. Finally, the same characterization is extended to infinite sets using the finite-basis properties of Klarner systems~\cite{Reid2005KlarnerSystems}. While in the finite case asymptotic cofiniteness is decidable, in the infinite case decidability may depend on the definition of the tile set. 

The paper is organized as follows. Section~\ref{s-preliminaries} recalls the
basic definitions of picture languages and tiling closures. Section~\ref{s-main}
contains the main theorem on asymptotic cofiniteness
for finite tile sets. Section~\ref{s-infinite} extends the result to arbitrary tile sets.
Section~\ref{s-concl} concludes with a discussion of related results and an open problem.

\section{Notations and preliminaries}\label{s-preliminaries}

A \emph{picture} is a rectangular array of letters over a finite alphabet. If a
picture has height \(h\) and width \(w\), its \emph{size} is \((h,w)\).
 Given two
pictures \(p,q\), their horizontal concatenation \(p\overt q\) is defined when
they have the same height, and their vertical concatenation
\(p\ominus q\) is defined when they have the same width. These
operations extend naturally to picture languages.

A \emph{tiling} of a picture is a partition of its domain into rectangular blocks,
each block carrying a picture from a prescribed language. The \emph{tiling
    closure} \(L^{++}\) of a language \(L\) is the set of all pictures that admit 
a tiling by elements of \(L\)~\cite{DBLP:journals/tcs/Simplot99}. 
Intuitively, $L^{++}$ is the set  of all  pictures that can be exactly paved, without gaps or overlaps, with pictures in $L$.

The
\emph{guillotine closure} \(L^{\ominus\overt +}\) is the smallest language
containing \(L\) and closed under both \(\overt\) and \(\ominus\). 

If \(q\in L^{\ominus\overt +}\), then \(q\) admits a partition \(r\) into
rectangular blocks from \(L\) such that \(q\) can be obtained by horizontal and
vertical concatenation of these blocks. Such a partition is called a
\emph{guillotine} (or \emph{sliceable}) decomposition of \(q\) into elements of
\(L\). Every guillotine decomposition is a tiling.

Equivalently, a guillotine decomposition can be described top-down by
recursively splitting a rectangle into two rectangles along a full-width or
full-height line. We call each such split a \emph{guillotine cut}. An example
is shown in Figure~\ref{fig-guillotine}.

\begin{figure}

\begin{center}
    \scalebox{0.6}{
        \begin{tikzpicture}[x=1cm,y=1cm,line cap=round]
            \colorlet{blockA}{gray!18}
            \colorlet{blockB}{gray!8}
            \colorlet{blockC}{gray!1}
            \colorlet{blockD}{gray!28}
            \colorlet{blockE}{gray!42}

            \foreach \r/\y in {1/0,2/-1,3/-4,4/-5}{
    \foreach \c/\x in {1/0,2/2.0,3/5.0,4/7,5/8}{
        \coordinate (\r\c) at (\x,\y);
    }
}
            
            \fill[blockD] (11) rectangle (42); 
            \fill[blockA] (12) rectangle (24); 
            \fill[blockE] (22) rectangle (33); 
            \fill[blockC] (32) rectangle (43); 
            \fill[blockB] (23) rectangle (44); 
            \fill[blockA] (14) rectangle (45);

            \draw (11) -- (41);
            \draw (12) -- (42);
            \draw (23) -- (43);
            \draw (14) -- (44);
            \draw (15) -- (45);
            
            \draw (11) -- (15);
            \draw (22) -- (24);
            \draw (32) -- (33);
            \draw (41) -- (45);

        \end{tikzpicture}
    }
\end{center}
\caption{This tiling is a guillotine decomposition: it admits a recursive decomposition by full-width or full-height guillotine cuts.}\label{fig-guillotine}
\end{figure}
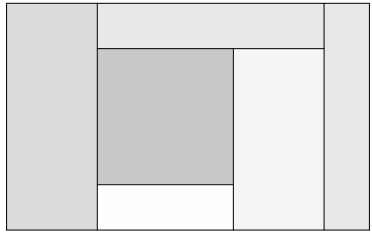

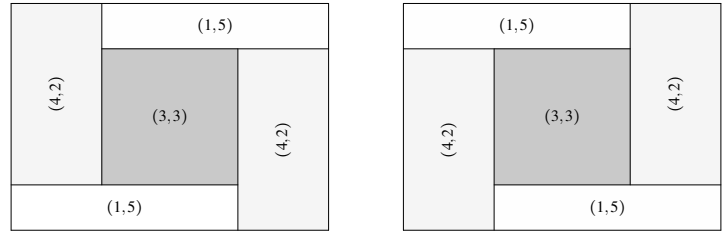
\begin{figure}
    
\begin{center}
    \scalebox{0.6}{
        \begin{tikzpicture}[x=1cm,y=1cm,line cap=round]
            \colorlet{blockB}{gray!8}
            \colorlet{blockC}{gray!1}
            \colorlet{blockD}{gray!8}
            \colorlet{blockE}{gray!42} 
            
            \foreach \r/\y in {1/0,2/-1,3/-4,4/-5}{
                \foreach \c/\x in {1/0,2/2.0,3/5.0,4/7}{
                    \coordinate (\r\c) at (\x,\y);
                }
            }
            
            \fill[blockC] (12) rectangle (24); 
            \fill[blockB] (23) rectangle (44); 
            \fill[blockC] (31) rectangle (33); 
            \fill[blockD] (11) rectangle (32); 
            \fill[blockE] (22) rectangle (33); 
            
            \node[font=\small] at ($(12)!0.5!(24)$) {$(1,5)$};
            \node[font=\small, rotate=90] at ($(23)!0.5!(44)$) {$(4,2)$};
            \node[font=\small] at ($(31)!0.5!(43)$) {$(1,5)$};
            \node[font=\small, rotate=90] at ($(11)!0.5!(32)$) {$(4,2)$};
            \node[font=\small] at ($(22)!0.5!(33)$) {$(3,3)$};
            
            \draw (11) -- (41);
            \draw (12) -- (32);
            \draw (23) -- (43);
            \draw (14) -- (44);
            
            \draw (11) -- (14);
            \draw (22) -- (24);
            \draw (31) -- (33);
            \draw (41) -- (44);
        \end{tikzpicture}
    }
    \qquad
    \scalebox{0.6}{
        \begin{tikzpicture}[x=1cm,y=1cm,line cap=round]
            \colorlet{blockB}{gray!8}
            \colorlet{blockC}{gray!1}
            \colorlet{blockD}{gray!8}
            \colorlet{blockE}{gray!42} 
            
            \foreach \r/\y in {1/0,2/-1,3/-4,4/-5}{
                \foreach \c/\x in {1/0,2/2.0,3/5.0,4/7}{
                    \coordinate (\r\c) at (\x,\y);
                }
            }
            
            \fill[blockC] (11) rectangle (23); 
            \fill[blockD] (13) rectangle (34); 
            \fill[blockC] (32) rectangle (44); 
            \fill[blockB] (21) rectangle (42); 
            \fill[blockE] (22) rectangle (33); 
            
            \node[font=\small] at ($(11)!0.5!(23)$) {$(1,5)$};
            \node[font=\small, rotate=90] at ($(13)!0.5!(34)$) {$(4,2)$};
            \node[font=\small] at ($(32)!0.5!(44)$) {$(1,5)$};
            \node[font=\small, rotate=90] at ($(21)!0.5!(42)$) {$(4,2)$};
            \node[font=\small] at ($(22)!0.5!(33)$) {$(3,3)$};
            
            \draw (11) -- (41);
            \draw (22) -- (42);
            \draw (13) -- (33);
            \draw (14) -- (44);
            
            \draw (11) -- (14);
            \draw (21) -- (23);
            \draw (32) -- (34);
            \draw (41) -- (44);
        \end{tikzpicture}
    }
\end{center}
\caption{The two displayed tilings over the tile set
    \(F=\{(1,5),(4,2),(3,3)\}\) are examples of so-called \emph{windmills}.
    Neither admits a guillotine decomposition into tiles of sizes in \(F\).}
\label{fig:windmills}
\end{figure}

The  inclusions  $
L \subseteq L^{\ominus\overt +}\subseteq   L^{++}$
hold for any language $L$. The second inclusion may be proper, as shown in
Figure~\ref{fig:windmills}. The two windmill pictures of size \((5,7)\)
admit the displayed tilings but no guillotine decomposition into tiles of
sizes in \(F=\{(1,5),(4,2),(3,3)\}\).

\paragraph{The unary case.}
From now on we work over a one-letter alphabet, say \(\{a\}\). Then a picture is
completely determined by its size; we write \(a^{h,w}\) for the unary picture of
height \(h\) and width \(w\). Thus a finite set
\(F\subseteq\mathbb N_{>0}^2\) represents the finite unary picture language
\[
\{\,a^{h,w} : (h,w)\in F\,\}.
\]

Under this identification, horizontal and vertical concatenation induce the
partial operations on sizes
\[
(h,w_1)\overt (h,w_2):=(h,w_1+w_2),
\qquad
(h_1,w)\ominus (h_2,w):=(h_1+h_2,w),
\]
whenever the dimensions match.

We write \(F^{\ominus\overt +}\) for the set of sizes of pictures in the
guillotine closure of the unary language represented by \(F\). Equivalently,
\(F^{\ominus\overt +}\) is the smallest subset of \(\mathbb N_{>0}^2\) containing
\(F\) and closed under the two operations above.

Similarly, we write \(F^{++}\) for the set of sizes of pictures in the tiling
closure of the unary language represented by \(F\). Thus \((h,w)\in F^{++}\) if
and only if the unary picture \(a^{h,w}\) admits a tiling by unary pictures whose
sizes belong to \(F\). Clearly\footnote{Both \(F^{++}\) and \(F^{\ominus\overt +}\) are Klarner systems~\cite{Reid2005KlarnerSystems}, i.e., subsets of \(\mathbb N_{>0}^2\) closed under horizontal and vertical composition. See Definition~\ref{def:Klarner} in Section~\ref{s-infinite}.}
\[
F^{\ominus\overt +}\subseteq F^{++}.
\]

\begin{definition}[asymptotic cofiniteness]
    
A subset \(S\subseteq \mathbb N_{>0}^2\) is said to be \emph{asymptotically cofinite} if there exists \((h_0,w_0)\in\mathbb N_{>0}^2\) such that
\[
(h,w)\in S \qquad\text{for all } h\ge h_0,\ w\ge w_0.
\]
\end{definition}

For an arbitrary nonempty set \(A\subseteq\mathbb N_{>0}\), we write
\(\gcd(A)\) for the greatest common divisor of all elements of \(A\), i.e., the
largest integer dividing every element of \(A\). We let \(\langle A\rangle\) denote
the additive semigroup generated by \(A\).

We use the conventions \(\gcd(\emptyset)=0\) and
\(\langle\emptyset\rangle=\emptyset\).

The following result is standard.

\begin{lemma}
    \label{lem:semigroup-cofinite}
    Let \(S\subseteq\mathbb N_{>0}\) be a nonempty additive semigroup. Then \(S\) is
    cofinite in \(\mathbb N_{>0}\) if and only if \(\gcd(S)=1\).
\end{lemma}
Indeed, if \(\gcd(S)=1\), then finitely many elements of \(S\) already have gcd \(1\);
the semigroup generated by them is a numerical semigroup and is cofinite, hence
so is \(S\).

In particular, if \(A\subseteq\mathbb N_{>0}\) is finite and nonempty, then
\(\langle A\rangle\) is cofinite if and only if \(\gcd(A)=1\).

\section{The finite case}\label{s-main}

We prove that both closures share the same arithmetic characterization.
Throughout this section, \(F\subseteq\mathbb N_{>0}^2\) is a nonempty finite set of tiles.

\subsection{Sufficiency for the guillotine closure}

Recall that
\[
H=\{\,h\ge 1 : \exists w\ge 1,\ (h,w)\in F\,\},
\qquad
P=\{\,p : p\text{ is prime and }p\mid w\text{ for some }(u,w)\in F\,\}.
\]
For each \(d\ge 2\), let
\[
H_d:=\{\,h\ge 1\mid \exists (h,w)\in F,\ d\nmid w\,\}.
\]

Notice that, for every prime \(p\notin P\), one has \(H_p=H\).

For each height \(h\ge 1\), set
\[
W_h:=\{\,w\ge 1\mid (h,w)\in F^{\ominus\overt +}\,\}.
\]
Since the guillotine
closure \(F^{\ominus\overt +}\) of the set $F$ of tiles is closed under horizontal concatenation, \(W_h\) is either empty or
an additive semigroup of \(\mathbb N_{>0}\).

We first prove the key lemma.

\begin{lemma}[Prime-avoiding LCM construction]
    \label{lem:Wh-cofinite}
    For every prime \(r\) and every \(h\in\langle H_r\rangle\), there exists
    \(w\in W_h\) such that \(r\nmid w\).
\end{lemma}

\begin{proof}
    Since \(h\in\langle H_r\rangle\), we may write
    \[
    h=h_1+\cdots+h_m
    \qquad (h_j\in H_r).
    \]
    For each \(j\), choose a tile \((h_j,w_j)\in F\) with \(r\nmid w_j\) (which is possible since $h_j\in H_r$), and let
    \[
    \ell=\operatorname{lcm}(w_1,\dots,w_m).
    \]
    Since \(r\) is prime and none of the \(w_j\) is divisible by \(r\), neither is
    \(\ell\). Since $\ell$ is a multiple of $w_j$, by horizontally concatenating \(\ell/w_j\) copies of \((h_j,w_j)\), we
    obtain \((h_j,\ell)\in F^{\ominus\overt +}\) for each \(j\). Vertically stacking
    these rectangles yields
    \[
    (h,\ell)\in F^{\ominus\overt +}.
    \]
    Thus \(\ell\in W_h\) and \(r\nmid \ell\).
\end{proof}

\begin{corollary}\label{cor:Wh-cofinite}
      If
  \[
  h\in \langle H\rangle \cap \bigcap_{p\in P}\langle H_p\rangle,
  \]
  then \(W_h\) is cofinite in \(\mathbb N_{>0}\).    
\end{corollary}
\begin{proof}
    Let
\[
h\in \langle H\rangle \cap \bigcap_{p\in P}\langle H_p\rangle.
\]
Let \(r\) be any prime. If \(r\in P\), then \(h\in\langle H_r\rangle\). If
\(r\notin P\), then \(H_r=H\), so again \(h\in\langle H_r\rangle\). By
Lemma~\ref{lem:Wh-cofinite}, there exists \(w\in W_h\) with \(r\nmid w\).
Therefore no prime divides every element of \(W_h\), hence \(\gcd(W_h)=1\).
Since \(W_h\) is a nonempty additive semigroup, Lemma~\ref{lem:semigroup-cofinite}
implies that \(W_h\) is cofinite.

\end{proof}
    
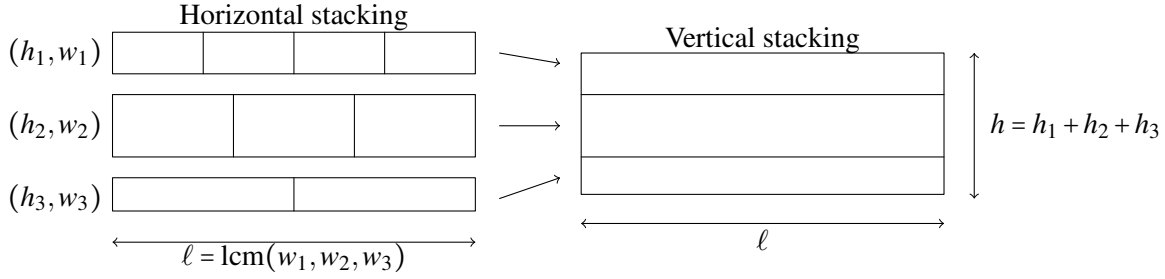
\begin{figure}[t]
    \centering
    \begin{tikzpicture}[x=0.4cm,y=0.275cm,line cap=round,line join=round]
        \def\L{12}
        \def\hA{2}
        \def\hB{3}
        \def\hC{2}
        \def\wA{3}
        \def\wB{4}
        \def\wC{6}
        
        \pgfmathsetmacro{\Xstack}{\L+3.5}
        \pgfmathsetmacro{\Xstackright}{\Xstack+\L}
        \pgfmathsetmacro{\Htot}{\hA+\hB+\hC}
        
        \node[left] at (0,-1) {$\,(h_1,w_1)$};
        \node[left] at (0,-4.5) {$\,(h_2,w_2)$};
        \node[left] at (0,-8) {$\,(h_3,w_3)$};
        
        \draw (0,0) rectangle (\L,-\hA);
        \foreach \x in {\wA,2*\wA,3*\wA} {
            \draw (\x,0) -- (\x,-\hA);
        }
        \node at (\L/2,0.7) {Horizontal stacking};
        
        \draw[->] (\L+0.8,-1) -- (\L+2.8,-1.5);
        
        \draw (0,-3) rectangle (\L,-3-\hB);
        \foreach \x in {\wB,2*\wB} {
            \draw (\x,-3) -- (\x,-3-\hB);
        }
        
        \draw[->] (\L+0.8,-4.5) -- (\L+2.8,-4.5);
        
        \draw[->] (\L+0.8,-8) -- (\L+2.8,-7);
        
        \draw (0,-7) rectangle (\L,-7-\hC*0.8);
        \draw (\wC,-7) -- (\wC,-7-\hC*0.8);
        
        \draw[<->] (0,-10.1) -- (\L,-10.1);
        \node at (\L/2,-10.85) {$\ell=\operatorname{lcm}(w_1,w_2,w_3)$};
        
        \draw (\Xstack,-1) rectangle (\Xstackright,-\Htot-0.8);
        \draw (\Xstack,-\hA-1) -- (\Xstackright,-\hA-1);
        \draw (\Xstack,-\hA-\hB-1) -- (\Xstackright,-\hA-\hB-1);
        \node at (\Xstack+\L/2,-0.4) {Vertical stacking};
        
        \draw[<->] (\Xstackright+1.0,-1) -- (\Xstackright+1.0,-\Htot-1);
        \node[right] at (\Xstackright+1.2,-\Htot/2-1) {$h=h_1+h_2+h_3$};
        
        \draw[<->] (\Xstack,-9.2) -- (\Xstackright,-9.2);
        \node at (\Xstack+\L/2,-10) {$\ell$};

    \end{tikzpicture}
    \caption{Construction used in Lemma~\ref{lem:Wh-cofinite}: if \(h=h_1+h_2+h_3\), each chosen tile \((h_i,w_i)\) is horizontally replicated \(\ell/w_i\) times to the common width \(\ell=\operatorname{lcm}(w_1,w_2,w_3)\), and the resulting rectangles are vertically stacked to obtain a picture in \(F^{\ominus\overt +}\) of size \((h,\ell)\).}
    \label{fig:stacking-construction}
\end{figure}

The sufficiency direction is now immediate.

\begin{theorem}[Sufficiency for the guillotine closure]
    \label{thm:guillotine-suff}
    If \(\gcd(H)=1\) and \(\gcd(H_p)=1\) for every \(p\in P\), then \(F^{\ominus\overt +}\) is asymptotically cofinite.
\end{theorem}

\begin{proof}
    By Lemma~\ref{lem:semigroup-cofinite}, the semigroup \(\langle H\rangle\) and each
    semigroup \(\langle H_p\rangle\) are cofinite in \(\mathbb N_{>0}\). Since \(P\) is
    finite, the set
    \[
    I:=\langle H\rangle \cap \bigcap_{p\in P}\langle H_p\rangle
    \]
    is cofinite in \(\mathbb N_{>0}\). For every \(h\in I\), Corollary~\ref{cor:Wh-cofinite}
    shows that \(W_h\) is cofinite.
    
    Choose two consecutive heights \(h_1,h_2\in I\). For each \(i\in\{1,2\}\), choose
    \(w_i\) such that 
    \[
    w\in W_{h_i} \text{ for all } w \ge w_i, 
    \]
    and set
    \[
    w_0=\max\{w_1,w_2\}.
    \]
    Then \((h_1,w),(h_2,w)\in F^{\ominus\overt +}\) for all \(w\ge w_0\). Since \(\gcd(h_1,h_2)=1\), there exists \(h_0\) such that every \(h\ge h_0\)
    belongs to \(\langle h_1,h_2\rangle\). Then, every sufficiently large \(h\) can be written as
    \[
    h=ah_1+bh_2
    \qquad (a,b\ge 0).
    \]
    Vertically stacking \(a\) copies of \((h_1,w)\) and \(b\) copies of \((h_2,w)\)
    yields \((h,w)\in F^{\ominus\overt +}\) for all sufficiently large \(h\) and all \(w\ge w_0\).
    Hence \(F^{\ominus\overt +}\) is asymptotically cofinite.
\end{proof}

\subsection{Necessity: a common proof for both closures}

\begin{lemma}[Roots-of-unity obstruction]
    \label{lem:root-obstruction}
    Let \(d\ge 2\) and let \(F\subseteq \mathbb N_{>0}^2\) be an arbitrary set of tiles. Define
    \[
    H_d:=\{h\ge 1:\exists (h,w)\in F,\ d\nmid w\}.
    \]
    If \(H_d=\emptyset\), then every \((h,w)\in F^{++}\) satisfies
    \(d\mid w\).
    
    \noindent If \(H_d\neq\emptyset\), then every \((h,w)\in F^{++}\) satisfies
    \[
    d\mid w
    \quad\text{or}\quad
    \gcd(H_d)\mid h.
    \]
\end{lemma}
\begin{proof}
      If \(H_d=\emptyset\), then every tile width is divisible by \(d\), and hence
every rectangle tiled by \(F\) has width divisible by \(d\).
Indeed, fixing any row of the tiled rectangle, the tiles meeting that row
partition the total width into tile widths, all divisible by \(d\).
 
Now assume that \(H_d\neq\emptyset\), and let \(g:=\gcd(H_d)\).
If \(g=1\), the conclusion is trivial.
So suppose \(g>1\).

    In a tiling, an occurrence \(T\) of a tile \((h_t,w_t)\in F\) is a concrete
rectangular block of cells \(T=R\times C\), where \(R\) is the interval of
row indices occupied by \(T\) and \(C\) is the interval of column indices
occupied by \(T\) (see Figure~\ref{fig:placed-tile}). Thus \(R\) and \(C\) are intervals of consecutive integers
with \(|R|=h_t\) and \(|C|=w_t\). 

The proof uses a complex-valued invariant based on roots of unity (see Wagon’s exposition of rectangle‑tiling arguments~\cite{Wagon1987}).

Let \(\omega=e^{2\pi i/g}\) and \(\zeta=e^{2\pi i/d}\), and define
\[
\Phi(T)=\sum_{(r,c)\in R\times C}\omega^r\zeta^c
=\Bigl(\sum_{r\in R}\omega^r\Bigr)\Bigl(\sum_{c\in C}\zeta^c\Bigr).
\]
\begin{description}
    \item[Case \(d\mid w_t\)] 
    
    If \(d\mid w_t\), then \(\sum_{c\in C}\zeta^c=0\), since \(|C|\) is divisible
    by \(d\) and the sum runs over a complete number of periods of the \(d\)-th
    roots of unity. Indeed, if \(C=\{b+1,\dots,b+w_t\}\), then
    \[
    \sum_{c\in C}\zeta^c
    =\zeta^{b+1}\sum_{s=0}^{w_t-1}\zeta^s
    =\zeta^{b+1}\frac{1-\zeta^{w_t}}{1-\zeta},
    \]
    which is zero whenever \(d\mid w_t\), because then \(\zeta^{w_t}=1\).
    
    \item[Case \(d\nmid w_t\)]
    If \(d\nmid w_t\), then \(h_t\in H_d\) and hence \(g\mid h_t\), since $g$ is the $\gcd$ of $H_d$. If
    \(R=\{a+1,\dots,a+h_t\}\), then
    \[
    \sum_{r\in R}\omega^r
    =\omega^{a+1}\sum_{s=0}^{h_t-1}\omega^s
    =\omega^{a+1}\frac{1-\omega^{h_t}}{1-\omega},
    \]
    which is zero whenever \(g\mid h_t\), because then \(\omega^{h_t}=1\).
\end{description}
In either case \(\Phi(T)=0\).

Now take any \((h,w)\in  F^{++}\) and a tiling of
\(\{1,\dots,h\}\times\{1,\dots,w\}\)
by tiles \(T_1,\dots,T_m\). Summing \(\Phi\) over the tiling gives
\[
\sum_{j=1}^m\Phi(T_j)
=
\sum_{r=1}^{h}\sum_{c=1}^{w}\omega^r\zeta^c,
\]
since the tiles partition the rectangle. Therefore,
\[
0
=
\sum_{j=1}^m\Phi(T_j)
=
\sum_{r=1}^{h}\sum_{c=1}^{w}\omega^r\zeta^c
=
\Bigl(\sum_{r=1}^{h}\omega^r\Bigr)
\Bigl(\sum_{c=1}^{w}\zeta^c\Bigr).
\]
Since \(\omega\) and \(\zeta\) are primitive roots of orders \(g\) and \(d\),
respectively, one has
\[
\sum_{r=1}^{h}\omega^r=0 \iff g\mid h,
\qquad
\sum_{c=1}^{w}\zeta^c=0 \iff d\mid w.
\]

Hence every rectangle in \(F^{++}\) satisfies

\begin{equation}\label{eq:divide}
    d\mid w \quad\text{or}\quad g\mid h.
\end{equation}

\end{proof}

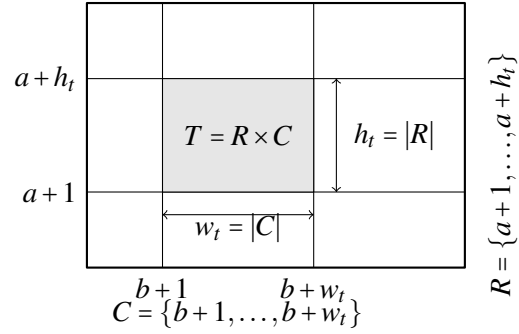
\begin{figure}[t]
    \centering
    \begin{tikzpicture}[x=0.5cm,y=0.5cm,line cap=round,line join=round]
        \draw[thick] (0,0) rectangle (10,7);
        
        \draw (2,0) -- (2,7);
        \draw (6,0) -- (6,7);
        \draw (0,2) -- (10,2);
        \draw (0,5) -- (10,5);
        
        \filldraw[fill=gray!20] (2,2) rectangle (6,5);
        \node at (4,3.5) {$T=R\times C$};
        
        \draw[<->] (2,1.4) -- (6,1.4);
        \node at (4,1.05) {$w_t=|C|$};
        
        \draw[<->] (6.6,2) -- (6.6,5);
        \node[right] at (6.8,3.5) {$h_t=|R|$};
        
        \node[below] at (2,0) {$b+1$};
        \node[below] at (6,0) {$b+w_t$};
        \node[left] at (0,2) {$a+1$};
        \node[left] at (0,5) {$a+h_t$};
        
        \node[below] at (4,-0.55) {$C=\{b+1,\dots,b+w_t\}$};
        \node[left,rotate=90] at (11,6) {$R=\{a+1,\dots,a+h_t\}$};
    \end{tikzpicture}
    \caption{A tile \(T\) placed in a tiling. The interval \(R\) records the row indices occupied by \(T\), and \(C\) records the column
        indices. In the proof of Lemma~\ref{lem:root-obstruction}, for a fixed divisor \(d\ge 2\) and \(g=\gcd(H_d)\), one associates with
        \(T\) the quantity
        \(\displaystyle
        \Phi(T)=\sum_{(r,c)\in R\times C}\omega^r\zeta^c.
        \)
        If \(d\mid w_t\), the sum over \(C\) vanishes; otherwise \(h_t\in H_d\), so
        \(g\mid h_t\) and the sum over \(R\) vanishes.}
    \label{fig:placed-tile}
\end{figure}

Let \(S\) be either the guillotine closure \(F^{\ominus\overt +}\) or the tiling closure
\(F^{++}\) of a finite tile set $F$. The argument below uses only the fact that every element of \(S\)
admits a tiling by tiles from \(F\); this holds for \(F^{++}\) by definition and for
\(F^{\ominus\overt +}\) because every guillotine decomposition is a tiling. Therefore the necessity
is established for both closures simultaneously.

\begin{theorem}[Necessity]
    \label{thm:necessity}
    Let \(S\in\{F^{\ominus\overt +},F^{++}\}\). If \(S\) is asymptotically cofinite, then
    \[
    \gcd(H)=1
    \qquad\text{and}\qquad
    \gcd(H_p)=1 \text{ for every } p\in P.
    \]
\end{theorem}

\begin{proof}
    Assume that \(S\) is asymptotically cofinite.
     If \(\gcd(H)>1\), then every tile has height divisible by \(\gcd(H)\).
    Thus every picture in \(S\) has height
    divisible by \(\gcd(H)\), contradicting asymptotic cofiniteness. Hence \(\gcd(H)=1\).
    
    Suppose by contradiction that there exists \(p\in P\) such that
    \(g:=\gcd(H_p)\neq 1\).
    
    If \(H_p=\emptyset\), then every tile has width divisible by \(p\). Thus, every picture in \(S\) has width
    divisible by \(p\), contradicting asymptotic cofiniteness. Thus \(H_p\neq\emptyset\), and
    $g>1$.
    
    By Lemma~\ref{lem:root-obstruction}, every element \((h,w)\in S\) satisfies
    \[
    p\mid w \quad\text{or}\quad g\mid h.
    \]
   
    Given any threshold \(h_0,w_0\), choose \(h\ge h_0\) with \(g\nmid h\) and
    \(w\ge w_0\) with \(p\nmid w\). Then \((h,w)\notin S\). Hence \(S\) is not
    asymptotically cofinite, a contradiction. Therefore, \(\gcd(H_p)=1\).
\end{proof}

\begin{remark}
    For the guillotine closure alone, the necessity of the prime conditions can be
    proved by a simple induction on the decomposition tree. Indeed, if for some
    \(p\in P\) one has \(H_p\neq\emptyset\) and \(g:=\gcd(H_p)>1\), then every tile
    \((h,w)\in F\) satisfies
    \[
    p\mid w \quad\text{or}\quad g\mid h.
    \]
    This condition is preserved by both horizontal and vertical concatenation: in a
    horizontal concatenation the height is unchanged and the widths add, while in a
    vertical concatenation the width is unchanged and the heights add. Hence the same divisibility condition holds for every picture in
    \(F^{\ominus\overt +}\).
    
    This tree-induction argument is not available for the full tiling closure, where
    a tiling need not admit a first full-width or full-height cut, as in windmill
    tilings. The roots-of-unity argument above replaces the missing structural
    induction by an additive invariant over all placed tiles.
\end{remark}

\iflong

\subsection{The final characterization}

Combining Theorems~\ref{thm:guillotine-suff} and \ref{thm:necessity}, and using
\(F^{\ominus\overt +}\subseteq F^{++}\), we obtain the full result.

\begin{corollary}[Unified characterization]
	\label{cor:unified}
	For a finite tile set \(F\subseteq\mathbb N_{>0}^2\), the following are equivalent:
	\begin{enumerate}
		\item \(F^{\ominus\overt +}\) is asymptotically cofinite;
		\item \(F^{++}\) is asymptotically cofinite;
		\item \(\displaystyle\gcd(H)=1\) and \(\gcd(H_{p})=1\) for every prime \(p\) dividing
		some tile width.
	\end{enumerate}
\end{corollary}

\begin{proof}
	\((3)\Rightarrow(1)\) is Theorem~\ref{thm:guillotine-suff}. Since
	\(F^{\ominus\overt +}\subseteq F^{++}\), this implies \((1)\Rightarrow(2)\).
	Finally, Theorem~\ref{thm:necessity} gives both \((1)\Rightarrow(3)\) and
	\((2)\Rightarrow(3)\).
\end{proof}

Although the condition is stated asymmetrically, it already forces the
one-dimensional width gcd condition: if all tile widths were divisible by a
prime \(p\), then \(p\in P\) and \(H_p=\emptyset\), contradicting
\(\gcd(H_p)=1\). The transposed condition obtained by exchanging height and
width is therefore equivalent by symmetry.

By Corollary~\ref{cor:unified}, the guillotine and tiling closures are asymptotically cofinite under exactly the same arithmetic condition. Since
\(
F^{\ominus\overt +}\subseteq F^{++},
\)
it follows that whenever this condition holds, both closures contain all sizes larger than some common rectangle; in particular, they are equal there.
Therefore, if the two closures are asymptotically cofinite, then for all sufficiently large unary pictures, the ability to tile a rectangle is no more powerful than the ability to build it by simple guillotine concatenations. 

\section{Extension to arbitrary tile sets}
\label{s-infinite}

We now extend the cofiniteness result to the case where the tile set  may be infinite.

We recall the definition of a Klarner system. 

\begin{definition}[Klarner system]\label{def:Klarner}
    A set \(S\subseteq \mathbb N_{>0}^2\) is a \emph{Klarner system} if it is closed
    under horizontal and vertical composition: whenever the indicated rectangles
    belong to \(S\),
    \[
    (h,w_1),(h,w_2)\in S
    \quad\Longrightarrow\quad
    (h,w_1+w_2)\in S,
    \]
    and
    \[
    (h_1,w),(h_2,w)\in S
    \quad\Longrightarrow\quad
    (h_1+h_2,w)\in S.
    \]
\end{definition}

Equivalently, \(S\) is a set of rectangle sizes closed under the two partial
operations \(\overt\) and \(\ominus\).

A picture language \(L\) is \emph{\(T\)-closed} if \(L=L^{++}\), and
\emph{\(C\)-closed} if \(L=L^{\ominus\overt +}\).
In the unary case, we identify a picture with its size
\((h,w)\in\mathbb N_{>0}^2\). 
 Under this identification, every \(C\)-closed
language is a Klarner system. Since every \(T\)-closed language is also
\(C\)-closed, unary \(T\)-closed languages are Klarner systems as well.

We shall use the finite-basis theorem for Klarner systems, e.g.,
Theorem~2.7 of~\cite{Reid2005KlarnerSystems}. An element \(x\) of
a Klarner system \(S\) is called \emph{prime} if \(S\setminus\{x\}\) is still
a Klarner system, equivalently if \(x\) cannot be obtained from elements of
\(S\setminus\{x\}\) by horizontal or vertical concatenation. We let
\(
\Pi(S)
\)
denote the set of prime elements of \(S\). The finite-basis theorem states that
\(\Pi(S)\) is finite and that it generates \(S\) by horizontal and vertical
concatenation.

\begin{proposition}[Finite Klarner reduction]
    \label{prop:finite-klarner-reduction}
    Let \(I\subseteq\mathbb N_{>0}^2\) be arbitrary, and let
    \[
    S_C:=I^{\ominus\overt +},
    \qquad
    S_T:=I^{++}.
    \]
    Then
    \[
    S_C=\Pi(S_C)^{\ominus\overt +}
    \qquad\text{and}\qquad
    S_T=\Pi(S_T)^{\ominus\overt +}.
    \]
\end{proposition}

\begin{proof}
    Both \(S_C\) and \(S_T\) are Klarner systems. By the finite-basis theorem
    for Klarner systems, \(\Pi(S_C)\) and \(\Pi(S_T)\) are finite and generate
    \(S_C\) and \(S_T\), respectively, by horizontal and vertical concatenation.
\end{proof}

Notice that \(\Pi(S_C)\subseteq I\). Indeed, if
\(x\in I^{\ominus\overt +}\setminus I\), then in a construction of \(x\) of
minimal length from \(I\), the last step is a nontrivial horizontal or
vertical concatenation. Hence \(x\) is not prime in the Klarner system
\(S_C\). For \(S_T=I^{++}\), the analogous inclusion need not hold: a prime
element of \(S_T\) may arise from a nonsliceable tiling by elements of \(I\).
However, every element of \(\Pi(S_T)\) is tileable by finitely many elements
of \(I\).

We now show that the characterization for infinite tile sets is the same as
in the finite case. We then briefly comment on decidability for infinite sets.

\begin{theorem}[Cofiniteness for arbitrary tile sets]
    \label{thm:cofinite-arbitrary}
    Let \(I\subseteq \mathbb N_{>0}^2\) be a nonempty, not necessarily finite, set of tiles.
    Let
    \[
    S_C:=I^{\ominus\overt +},
    \qquad
    S_T:=I^{++}.
    \]
    Let
    \[
    H:=\{h\ge 1:\exists w\ge 1,\ (h,w)\in I\},
    \]
    and let
    \[
    P:=\{p:\ p\text{ is prime and }p\mid w
    \text{ for some }(h,w)\in I\}.
    \]
    For \(p\in P\), define
    \[
    H_p:=\{h\ge 1:\exists (h,w)\in I,\ p\nmid w\}.
    \]
    Then the following are equivalent:
    \begin{enumerate}
        \item \(S_C\) is asymptotically cofinite;
        \item \(S_T\) is asymptotically cofinite;
        \item
        \[
        \gcd(H)=1
        \qquad\text{and}\qquad
        \gcd(H_p)=1 \text{ for every } p\in P.
        \]
    \end{enumerate}
\end{theorem}

\begin{proof}
    We first prove necessity. Since \(S_C\subseteq S_T\), it is enough to prove
    the necessary conditions assuming that \(S_T\) is asymptotically cofinite.
    
    If
    \(
    g:=\gcd(H)>1,
    \)
    then every tile height is divisible by \(g\). Hence every rectangle tiled by
    \(I\) has height divisible by \(g\), contradicting
    asymptotic cofiniteness. Therefore
    \(
    \gcd(H)=1.
    \)
    
    Now let \(p\in P\). If \(H_p=\emptyset\), then every tile width is divisible
    by \(p\). Hence every rectangle tiled by \(I\) has width divisible by \(p\),
    because the tiles meeting any fixed row partition the total width into tile
    widths. This again contradicts asymptotic cofiniteness. Thus
    \(H_p\neq\emptyset\).
    
    Suppose by contradiction that
    \(
    g_p:=\gcd(H_p)>1.
    \)
    By Lemma~\ref{lem:root-obstruction}, applied with \(d=p\), every
    \((h,w)\in S_T\) satisfies
    \[
    p\mid w
    \quad\text{or}\quad
    g_p\mid h.
    \]
    Therefore all rectangles \((h,w)\) with
    \(
    p\nmid w
    \qquad\text{and}\qquad
    g_p\nmid h
    \)
    are missing from \(S_T\). Since there are arbitrarily large such rectangles,
    \(S_T\) cannot be asymptotically cofinite. This contradiction gives
    \(
    \gcd(H_p)=1.
    \)
    
    To prove sufficiency, assume that the gcd conditions in (3) hold for \(I\). By
    Proposition~\ref{prop:finite-klarner-reduction}, the finite sets
    \[
    G_C:=\Pi(S_C),
    \qquad
    G_T:=\Pi(S_T)
    \]
    satisfy
    \[
    S_C=G_C^{\ominus\overt +},
    \qquad
    S_T=G_T^{\ominus\overt +}.
    \]
    
    For a finite set \(G\), let \(H^G\), \(H_p^G\), and \(P^G\) denote,
    respectively, the sets \(H\), \(H_p\), and \(P\) computed from \(G\)
    rather than from \(I\).
     We claim that \(G_C\) satisfies the finite gcd conditions of
    Corollary~\ref{cor:unified}, namely: 
    \[
    \gcd(H^{G_C})=1,\qquad \gcd(H^{G_C}_p)=1
    \quad(p\in P^{G_C}).
    \] If one of these conditions failed for \(G_C\),
    then the finite-case necessity proof, applied to the finite tile set \(G_C\), would give one of the following obstructions
    holding for every element of
    \(
    G_C^{\ominus\overt +}=S_C:
    \)
    \begin{enumerate}
        \item all heights are divisible by some integer \(g>1\);
        \item for some prime \(p\), all widths are divisible by \(p\);
        \item for some prime \(p\) and some integer \(g>1\), every
        \((h,w)\in S_C\) satisfies
        \(
        p\mid w
        \quad\text{or}\quad
        g\mid h.
        \)
    \end{enumerate}
    Since \(I\subseteq S_C\), the same obstruction holds for all tiles of
    \(I\). In case (1), this contradicts \(\gcd(H)=1\). In case (2), every
    width of a tile in \(I\) is divisible by \(p\); hence \(p\in P\) and
    \(H_p=\emptyset\), contradicting \(\gcd(H_p)=1\). In case (3), if
    \(p\notin P\), then no width of a tile in \(I\) is divisible by \(p\), so
    all tile heights are divisible by \(g\), contradicting \(\gcd(H)=1\). If
    \(p\in P\), then all heights in \(H_p\) are divisible by \(g\), contradicting
    \(\gcd(H_p)=1\). Thus \(G_C\) satisfies the finite gcd conditions.
    
    By Corollary~\ref{cor:unified},
    \(
    S_C=G_C^{\ominus\overt +}
    \)
    is asymptotically cofinite.
    
    The same argument applies to \(G_T\). 
    Indeed, if the finite gcd conditions failed for \(G_T\), the finite-case necessity proof applied to \(G_T\) would yield an obstruction holding for every element of \(G_T^{\ominus\overt +}=S_T\). 
    Since \(I\subseteq S_T\), the same obstruction would hold for all tiles of \(I\), contradicting the assumed gcd conditions for \(I\). Hence \(G_T\) also satisfies the finite gcd conditions.
    
    Corollary~\ref{cor:unified} gives that
    \(
    S_T=G_T^{\ominus\overt +}
    \)
    is asymptotically cofinite.
    
    This proves the equivalence.
\end{proof}

\begin{remark}
    For infinite \(I\), the sets \(H\) and \(H_p\) may themselves be infinite.
    Nevertheless, only the additive semigroups \(\langle H\rangle\) and
    \(\langle H_p\rangle\), for \(p\in P\), occur in criterion (3) of Theorem~\ref{thm:cofinite-arbitrary}. 
    These are additive
    subsemigroups of \(\mathbb N_{>0}\), hence finitely generated. In particular,
    \(\gcd(H)=1\), respectively \(\gcd(H_p)=1\), is equivalent to cofiniteness of
    \(\langle H\rangle\), respectively \(\langle H_p\rangle\), and is witnessed by
    finitely many elements of \(H\), respectively \(H_p\). What may remain infinite
    is the family of primes \(p\in P\) for which the condition has to be checked. 
    The criterion is effective only under
    additional assumptions ensuring that \(H\), \(P\), and the sets \(H_p\) can be
    effectively analyzed. For example, if \(I\) is given as a recognizable picture language, one may ask
    whether the sets \(H\), \(P\), and \(H_p\) are effectively analyzable, and hence
    whether the above gcd criterion becomes decidable. This is independent of the
    finite-basis reduction used here and is left open.

    Notice that the finite-basis theorem for Klarner systems is used in the proof only as an
    existence result. In general, the bases \(\Pi(S_C)\) and \(\Pi(S_T)\) are not
    provided constructively by the theorem, and need not be computable from an
    arbitrary description of an infinite tile set \(I\). Thus the argument reduces
    the mathematics to the finite case, but it does not by itself yield an effective
    procedure for checking cofiniteness. 
\end{remark}
\fi
\section{Conclusion}\label{s-concl}

We studied the cofiniteness problem for unary picture languages generated by
rectangular tiles. Two natural closures arise in this setting: the guillotine
closure, obtained by repeated horizontal and vertical concatenation, and the
tiling closure, obtained by arbitrary tilings. Our main result shows that these
two closures have the same asymptotic cofiniteness behaviour. More precisely,
the criterion of Corollary~\ref{cor:unified}, extended to arbitrary tile sets
in Theorem~\ref{thm:cofinite-arbitrary}, is necessary and sufficient for both
closures.

This result is closely related to the theory of Klarner systems. For a finite tile set \(F\), both \(F^{\ominus\overt+}\) and \(F^{++}\) define two-dimensional Klarner systems of rectangle sizes.
 Barnes' characterization theorem, in the form presented as
Theorem~3.3 in~\cite{Reid2008Barnes}, implies that every Klarner system is eventually
characterized by a finite set of divisibility restrictions. Consequently, for a
Klarner system, asymptotic cofiniteness is equivalent to the absence of
nontrivial restrictions. Our contribution is to make this criterion explicit for
the systems generated by a finite set of rectangular tiles, translating it into
simple finite arithmetic conditions. These
conditions were obtained here by a direct analysis of guillotine compositions
and tiling obstructions; nevertheless, in hindsight, Theorems~3.3 and~3.6
of~\cite{Reid2008Barnes} could also be used to give an alternative proof of the same
gcd characterization, once the absence of nontrivial restrictions is translated
into our arithmetic conditions.

The same viewpoint also clarifies the relation between guillotine and arbitrary
tiling closures beyond the cofinite case. By Barnes' preservation theorem for
restrictions (Theorem~3.6 of~\cite{Reid2008Barnes}) the two closures satisfy the same
divisibility restrictions. Hence, by Barnes' characterization theorem, their
membership conditions are eventually determined by the same finite set of
restrictions. It follows that \(F^{\ominus\overt+}\) and \(F^{++}\) coincide for
all sufficiently large rectangle sizes. Although we do not use this general eventual
equivalence in our proofs, it shows that the asymptotic
agreement of the two closures is not limited to the cofinite case.

A natural open problem is to extend the characterization to non-unary alphabets.
In one dimension, related cofiniteness questions have been considered for word
languages over general alphabets~\cite{DBLP:conf/stacs/KaoSX08}. For pictures,
the situation is more delicate: cofiniteness would require generating all
sufficiently large labeled pictures, not only all sufficiently large rectangle
sizes, and general tiling and universality problems are often undecidable.

\paragraph{Acknowledgements} We are grateful to the anonymous reviewers for their comments. In particular,
one reviewer pointed out the close connection with Reid's presentation of
Barnes' results, and another reviewer observed that the eventual equivalence of
the guillotine and tiling closures follows from those results.

\bibliographystyle{eptcs}
 \bibliography{automatabib}

\begin{thebibliography}{10}
\providecommand{\bibitemdeclare}[2]{}
\providecommand{\surnamestart}{}
\providecommand{\surnameend}{}
\providecommand{\urlprefix}{Available at }
\providecommand{\url}[1]{\texttt{#1}}
\providecommand{\href}[2]{\texttt{#2}}
\providecommand{\urlalt}[2]{\href{#1}{#2}}
\providecommand{\doi}[1]{doi:\urlalt{https://doi.org/#1}{#1}}
\providecommand{\eprint}[1]{arXiv:\urlalt{https://arxiv.org/abs/#1}{#1}}
\providecommand{\bibinfo}[2]{#2}

\bibitemdeclare{article}{DBLP:journals/combtheory/AsinowskiCFF25}
\bibitem{DBLP:journals/combtheory/AsinowskiCFF25}
\bibinfo{author}{Andrei \surnamestart Asinowski\surnameend},
  \bibinfo{author}{Jean \surnamestart Cardinal\surnameend},
  \bibinfo{author}{Stefan \surnamestart Felsner\surnameend} \&
  \bibinfo{author}{{\'{E}}ric \surnamestart Fusy\surnameend}
  (\bibinfo{year}{2025}): \emph{\bibinfo{title}{Combinatorics of
  rectangulations: old and new bijections}}.
\newblock {\slshape \bibinfo{journal}{Comb. Theory}}
  \bibinfo{volume}{5}(\bibinfo{number}{1}), \doi{10.5070/C65165025}.

\bibitemdeclare{inproceedings}{BMV2007}
\bibitem{BMV2007}
\bibinfo{author}{Alberto \surnamestart Bertoni\surnameend},
  \bibinfo{author}{Massimiliano \surnamestart Goldwurm\surnameend} \&
  \bibinfo{author}{Violetta \surnamestart Lonati\surnameend}
  (\bibinfo{year}{2007}): \emph{\bibinfo{title}{On the Complexity of Unary
  Tiling-Recognizable Picture Languages}}.
\newblock In \bibinfo{editor}{Wolfgang \surnamestart Thomas\surnameend} \&
  \bibinfo{editor}{Pascal \surnamestart Weil\surnameend}, editors: {\slshape
  \bibinfo{booktitle}{{STACS} 2007, 24th Annual Symposium on Theoretical
  Aspects of Computer Science, Aachen, Germany, February 22-24, 2007,
  Proceedings}}, \bibinfo{series}{Lecture Notes in Computer Science},
  \bibinfo{publisher}{Springer}, pp. \bibinfo{pages}{381--392},
  \doi{10.1007/978-3-540-70918-3\_33}.

\bibitemdeclare{article}{DBLP:journals/tcs/CherubiniCPP06}
\bibitem{DBLP:journals/tcs/CherubiniCPP06}
\bibinfo{author}{Alessandra \surnamestart Cherubini\surnameend},
  \bibinfo{author}{Stefano \surnamestart Crespi{-}Reghizzi\surnameend},
  \bibinfo{author}{Matteo \surnamestart Pradella\surnameend} \&
  \bibinfo{author}{Pierluigi \surnamestart {San Pietro}\surnameend}
  (\bibinfo{year}{2006}): \emph{\bibinfo{title}{Picture languages: Tiling
  systems versus tile rewriting grammars}}.
\newblock {\slshape \bibinfo{journal}{Theoret. Comput. Sci.}}
  \bibinfo{volume}{356}(\bibinfo{number}{1-2}), pp. \bibinfo{pages}{90--103},
  \doi{10.1016/j.tcs.2006.01.038}.

\bibitemdeclare{incollection}{DBLP:books/ems/21/Crespi-ReghizziGL21}
\bibitem{DBLP:books/ems/21/Crespi-ReghizziGL21}
\bibinfo{author}{Stefano \surnamestart {Crespi Reghizzi}\surnameend},
  \bibinfo{author}{Dora \surnamestart Giammarresi\surnameend} \&
  \bibinfo{author}{Violetta \surnamestart Lonati\surnameend}
  (\bibinfo{year}{2021}): \emph{\bibinfo{title}{Two-dimensional models}}.
\newblock In \bibinfo{editor}{Jean{-}{\'{E}}ric \surnamestart Pin\surnameend},
  editor: {\slshape \bibinfo{booktitle}{Handbook of Automata Theory}},
  \bibinfo{publisher}{European Mathematical Society Publishing House},
  \bibinfo{address}{Berlin}, pp. \bibinfo{pages}{303--333},
  \doi{10.4171/automata-1/9}.

\bibitemdeclare{inproceedings}{DBLP:conf/dlt/CrespiReghizziRP26}
\bibitem{DBLP:conf/dlt/CrespiReghizziRP26}
\bibinfo{author}{Stefano \surnamestart Crespi{-}Reghizzi\surnameend},
  \bibinfo{author}{Antonio \surnamestart Restivo\surnameend} \&
  \bibinfo{author}{Pierluigi \surnamestart {San Pietro}\surnameend}
  (\bibinfo{year}{2026}): \emph{\bibinfo{title}{Closure Operations on Picture
  Languages and Their Relation to Floor Plans}}.
\newblock In: {\slshape \bibinfo{booktitle}{{DLT}}}, {\slshape
  \bibinfo{series}{Lecture Notes in Computer Science}} \bibinfo{volume}{16578},
  \bibinfo{publisher}{Springer}, pp. \bibinfo{pages}{168--180},
  \doi{10.1007/978-3-032-28404-4_13}.

\bibitemdeclare{article}{DBLP:journals/siamcomp/EppsteinMSV12}
\bibitem{DBLP:journals/siamcomp/EppsteinMSV12}
\bibinfo{author}{David \surnamestart Eppstein\surnameend},
  \bibinfo{author}{Elena \surnamestart Mumford\surnameend},
  \bibinfo{author}{Bettina \surnamestart Speckmann\surnameend} \&
  \bibinfo{author}{Kevin \surnamestart Verbeek\surnameend}
  (\bibinfo{year}{2012}): \emph{\bibinfo{title}{Area-Universal and Constrained
  Rectangular Layouts}}.
\newblock {\slshape \bibinfo{journal}{{SIAM} J. Comput.}}
  \bibinfo{volume}{41}(\bibinfo{number}{3}), pp. \bibinfo{pages}{537--564},
  \doi{10.1137/110834032}.

\bibitemdeclare{article}{GiammarresiR92}
\bibitem{GiammarresiR92}
\bibinfo{author}{Dora \surnamestart Giammarresi\surnameend} \&
  \bibinfo{author}{Antonio \surnamestart Restivo\surnameend}
  (\bibinfo{year}{1992}): \emph{\bibinfo{title}{Recognizable Picture
  Languages}}.
\newblock {\slshape \bibinfo{journal}{Int. J. Pattern Recognit. Artif.
  Intell.}} \bibinfo{volume}{6}(\bibinfo{number}{2{\&}3}), pp.
  \bibinfo{pages}{241--256}, \doi{10.1142/S021800149200014X}.

\bibitemdeclare{incollection}{GiammRestivo1997}
\bibitem{GiammRestivo1997}
\bibinfo{author}{Dora \surnamestart Giammarresi\surnameend} \&
  \bibinfo{author}{Antonio \surnamestart Restivo\surnameend}
  (\bibinfo{year}{1997}): \emph{\bibinfo{title}{Two-dimensional languages}}.
\newblock In \bibinfo{editor}{G.~\surnamestart Rozenberg\surnameend} \&
  \bibinfo{editor}{A.~\surnamestart Salomaa\surnameend}, editors: {\slshape
  \bibinfo{booktitle}{Handbook of formal languages, vol. 3}},
  \bibinfo{publisher}{Springer}, \bibinfo{address}{Berlin}, pp.
  \bibinfo{pages}{215--267}, \doi{10.1007/978-3-642-59126-6_4}.

\bibitemdeclare{article}{iori2021exact}
\bibitem{iori2021exact}
\bibinfo{author}{Manuel \surnamestart Iori\surnameend},
  \bibinfo{author}{Vin{\'\i}cius~L. \surnamestart de~Lima\surnameend},
  \bibinfo{author}{Silvano \surnamestart Martello\surnameend},
  \bibinfo{author}{Fl{\'a}vio~K. \surnamestart Miyazawa\surnameend} \&
  \bibinfo{author}{Michele \surnamestart Monaci\surnameend}
  (\bibinfo{year}{2021}): \emph{\bibinfo{title}{Exact solution techniques for
  two-dimensional cutting and packing}}.
\newblock {\slshape \bibinfo{journal}{European J. Oper. Res.}}
  \bibinfo{volume}{289}(\bibinfo{number}{2}), pp. \bibinfo{pages}{399--415},
  \doi{10.1016/j.ejor.2020.06.050}.

\bibitemdeclare{article}{DBLP:journals/tcs/KantH97}
\bibitem{DBLP:journals/tcs/KantH97}
\bibinfo{author}{Goos \surnamestart Kant\surnameend} \& \bibinfo{author}{Xin
  \surnamestart He\surnameend} (\bibinfo{year}{1997}):
  \emph{\bibinfo{title}{Regular Edge Labeling of 4-Connected Plane Graphs and
  Its Applications in Graph Drawing Problems}}.
\newblock {\slshape \bibinfo{journal}{Theoret. Comput. Sci.}}
  \bibinfo{volume}{172}(\bibinfo{number}{1-2}), pp. \bibinfo{pages}{175--193},
  \doi{10.1016/S0304-3975(95)00257-X}.

\bibitemdeclare{inproceedings}{DBLP:conf/stacs/KaoSX08}
\bibitem{DBLP:conf/stacs/KaoSX08}
\bibinfo{author}{Jui{-}Yi \surnamestart Kao\surnameend},
  \bibinfo{author}{Jeffrey~O. \surnamestart Shallit\surnameend} \&
  \bibinfo{author}{Zhi \surnamestart Xu\surnameend} (\bibinfo{year}{2008}):
  \emph{\bibinfo{title}{The Frobenius Problem in a Free Monoid}}.
\newblock In \bibinfo{editor}{Susanne \surnamestart Albers\surnameend} \&
  \bibinfo{editor}{Pascal \surnamestart Weil\surnameend}, editors: {\slshape
  \bibinfo{booktitle}{Proceedings of the 25th Annual Symposium on Theoretical
  Aspects of Computer Science, {STACS} 2008, Bordeaux, France, February 21-23,
  2008}}, \bibinfo{series}{LIPIcs}, pp. \bibinfo{pages}{421--432},
  \doi{10.4230/LIPICS.STACS.2008.1362}.

\bibitemdeclare{article}{DBLP:journals/tcs/KumarS21}
\bibitem{DBLP:journals/tcs/KumarS21}
\bibinfo{author}{Vinod \surnamestart Kumar\surnameend} \&
  \bibinfo{author}{Krishnendra \surnamestart Shekhawat\surnameend}
  (\bibinfo{year}{2021}): \emph{\bibinfo{title}{A transformation algorithm to
  construct a rectangular floorplan}}.
\newblock {\slshape \bibinfo{journal}{Theoret. Comput. Sci.}}
  \bibinfo{volume}{871}, pp. \bibinfo{pages}{94--106},
  \doi{10.1016/J.TCS.2021.04.014}.

\bibitemdeclare{article}{DBLP:journals/tcs/LatteuxS97}
\bibitem{DBLP:journals/tcs/LatteuxS97}
\bibinfo{author}{Michel \surnamestart Latteux\surnameend} \&
  \bibinfo{author}{David \surnamestart Simplot\surnameend}
  (\bibinfo{year}{1997}): \emph{\bibinfo{title}{Recognizable Picture Languages
  and Domino Tiling}}.
\newblock {\slshape \bibinfo{journal}{Theoret. Comput. Sci.}}
  \bibinfo{volume}{178}(\bibinfo{number}{1-2}), pp. \bibinfo{pages}{275--283},
  \doi{10.1016/S0304-3975(96)00283-6}.

\bibitemdeclare{article}{DBLP:journals/dam/LodiMP17}
\bibitem{DBLP:journals/dam/LodiMP17}
\bibinfo{author}{Andrea \surnamestart Lodi\surnameend},
  \bibinfo{author}{Michele \surnamestart Monaci\surnameend} \&
  \bibinfo{author}{Enrico \surnamestart Pietrobuoni\surnameend}
  (\bibinfo{year}{2017}): \emph{\bibinfo{title}{Partial enumeration algorithms
  for Two-Dimensional Bin Packing Problem with guillotine constraints}}.
\newblock {\slshape \bibinfo{journal}{Discrete Appl. Math.}}
  \bibinfo{volume}{217}, pp. \bibinfo{pages}{40--47},
  \doi{10.1016/J.DAM.2015.09.012}.

\bibitemdeclare{inproceedings}{DBLP:conf/stacs/Matz97}
\bibitem{DBLP:conf/stacs/Matz97}
\bibinfo{author}{Oliver \surnamestart Matz\surnameend} (\bibinfo{year}{1997}):
  \emph{\bibinfo{title}{Regular Expressions and Context-Free Grammars for
  Picture Languages}}.
\newblock In \bibinfo{editor}{R{\"{u}}diger \surnamestart Reischuk\surnameend}
  \& \bibinfo{editor}{Michel \surnamestart Morvan\surnameend}, editors:
  {\slshape \bibinfo{booktitle}{STACS 97, 14th Annual Symposium on Theoretical
  Aspects of Computer Science, L{\"u}beck, Germany, February 27--March 1, 1997,
  Proceedings}}, {\slshape \bibinfo{series}{Lecture Notes in Computer Science}}
  \bibinfo{volume}{1200}, \bibinfo{publisher}{Springer}, pp.
  \bibinfo{pages}{283--294}, \doi{10.1007/BFB0023466}.

\bibitemdeclare{article}{DBLP:journals/iandc/PradellaCC11}
\bibitem{DBLP:journals/iandc/PradellaCC11}
\bibinfo{author}{Matteo \surnamestart Pradella\surnameend},
  \bibinfo{author}{Alessandra \surnamestart Cherubini\surnameend} \&
  \bibinfo{author}{Stefano \surnamestart Crespi{-}Reghizzi\surnameend}
  (\bibinfo{year}{2011}): \emph{\bibinfo{title}{A unifying approach to picture
  grammars}}.
\newblock {\slshape \bibinfo{journal}{Inform. and Comput.}}
  \bibinfo{volume}{209}(\bibinfo{number}{9}), pp. \bibinfo{pages}{1246--1267},
  \doi{10.1016/J.IC.2011.07.001}.

\bibitemdeclare{article}{Reid2005KlarnerSystems}
\bibitem{Reid2005KlarnerSystems}
\bibinfo{author}{Michael \surnamestart Reid\surnameend} (\bibinfo{year}{2005}):
  \emph{\bibinfo{title}{Klarner Systems and Tiling Boxes with Polyominoes}}.
\newblock {\slshape \bibinfo{journal}{J. Combin. Theory Ser. A}}
  \bibinfo{volume}{111}(\bibinfo{number}{1}), pp. \bibinfo{pages}{89--105},
  \doi{10.1016/j.jcta.2004.10.010}.

\bibitemdeclare{article}{Reid2008Barnes}
\bibitem{Reid2008Barnes}
\bibinfo{author}{Michael \surnamestart Reid\surnameend} (\bibinfo{year}{2008}):
  \emph{\bibinfo{title}{Asymptotically Optimal Box Packing Theorems}}.
\newblock {\slshape \bibinfo{journal}{Electron. J. Combin.}}
  \bibinfo{volume}{15}(\bibinfo{number}{1}), p. \bibinfo{pages}{R78},
  \doi{10.37236/802}.

\bibitemdeclare{article}{DBLP:journals/tcs/Simplot99}
\bibitem{DBLP:journals/tcs/Simplot99}
\bibinfo{author}{David \surnamestart Simplot\surnameend}
  (\bibinfo{year}{1999}): \emph{\bibinfo{title}{A Characterization of
  Recognizable Picture Languages by Tilings by Finite Sets}}.
\newblock {\slshape \bibinfo{journal}{Theoret. Comput. Sci.}}
  \bibinfo{volume}{218}(\bibinfo{number}{2}), pp. \bibinfo{pages}{297--323},
  \doi{10.1016/S0304-3975(98)00328-4}.

\bibitemdeclare{article}{Wagon1987}
\bibitem{Wagon1987}
\bibinfo{author}{Stan \surnamestart Wagon\surnameend} (\bibinfo{year}{1987}):
  \emph{\bibinfo{title}{Fourteen proofs of a result about tiling a rectangle}}.
\newblock {\slshape \bibinfo{journal}{Amer. Math. Monthly}}
  \bibinfo{volume}{94}(\bibinfo{number}{7}), pp. \bibinfo{pages}{601--617},
  \doi{10.2307/2322213}.

\end{thebibliography}

\end{document}